\documentclass[runningheads]{llncs}
\usepackage{verbatim}
\usepackage{graphicx}
\usepackage{xcolor}
\usepackage{fullpage}
\usepackage{amsmath}
\usepackage{amssymb}
\usepackage[ruled,vlined,linesnumbered]{algorithm2e}

\newcommand{\val}[2]{{v_{#1}(#2)}}

\newcommand{\alloc}[2]{x_{#1}(#2)}

\newcommand{\lowsig}{0}

\newcommand{\qual}{q}
\newcommand{\qVec}{\mathbf{q}}
\newcommand{\quality}[1]{\qual(#1)}

\newcommand{\sVec}{\mathbf{s}}
\newcommand{\pVec}{\mathbf{p}}
\newcommand{\mechanism}{\mathcal{A}}
\newcommand{\agents}{N}
\newcommand{\opta}[1]{h(#1)}

\newcommand{\nExperts}{\ell}

\newcommand{\clockStrategy}{{consistent bidding}}

\usepackage{appendix}

\begin{document}
\title{Prior-Free Clock Auctions for Bidders\\ with Interdependent Values\thanks{The first and last authors were partially supported by NSF grants CCF-2008280 and CCF-1755955. The second author was supported by an REU through CCF-1755955.}}
%
%
\author{
Vasilis Gkatzelis \and
Rishi Patel \and
Emmanouil Pountourakis \and
Daniel Schoepflin}
\authorrunning{V. Gkatzelis et al.}
%
\institute{Drexel University, Philadelphia, USA\\
\email{\{gkatz,riship,manolis,schoep\}@drexel.edu}\\}
\maketitle              
\begin{abstract}
We study the problem of selling a good to a group of bidders with interdependent values in a prior-free setting. Each bidder has a signal that can take one of $k$ different values, and her value for the good is a weakly increasing function of all the bidders' signals. The bidders are partitioned into $\nExperts$ expertise-groups, based on how their signal can impact the values for the good, and we prove upper and lower bounds regarding the approximability of social welfare and revenue for a variety of settings, parameterized by $k$ and $\nExperts$. Our lower bounds apply to all
ex-post incentive compatible mechanisms and our upper bounds are all within a small constant of the lower bounds. Our main results take the appealing form of ascending clock auctions and provide strong incentives by admitting the desired outcomes as \emph{obvious ex-post equilibria}.

\keywords{Clock auctions  \and Interdependent values \and Obvious ex-post equilibrium.}
\end{abstract}

\section{Introduction}
We study the problem of selling a good to bidders with \emph{interdependent values}, which has received a lot of attention in economics (e.g., see \cite[Chapters~6 and 10]{krishna2009auction}), and recently also in computer science (e.g., \cite{roughgarden2016optimal,chawla2014approximate,eden2019combinatorial,eden2018interdependent,eden2020price,constantin2007online,ito2006instantiating,robu2013efficient}). In contrast to the  \emph{private values} model, where each bidder knows her value for the good being sold, the interdependent value literature assumes that each bidder has some private signal regarding the value of the good, e.g., through some research or technical expertise, and the actual value of the good to each bidder is a function of all the bidders' signals. For instance, a common motivating example for this problem involves firms competing over the mineral rights of a piece of land \cite{wilson1969competitive}: each firm has conducted some tests, trying to estimate the land's capacity in desired minerals, but each of these tests may provide only partial evidence, and the best estimate can be inferred by appropriately aggregating all the test results, e.g., by computing the average across all of these measurements.

The main difficulty when designing auctions for bidders with interdependent values arises from the fact that the bidders' signals are not known to the auctioneer, or to the other bidders. Therefore, the auctioneer needs to elicit these signals before deciding who should win the item and what the price should be. But, why would any bidder reveal her true signal to the auctioneer? A sealed-bid auction is said to be \emph{ex-post incentive compatible} if truth-telling, i.e., reporting the true signal to the auctioneer, is an equilibrium for all the bidders. Designing ex-post incentive compatible auctions with non-trivial welfare or revenue guarantees has been a central goal of this line of research.

Prior work has considered several different ways in which the bidders' values can depend on the vector of signals. For example, in the \emph{common value} model all the bidders have the same value for the good but, even in this special case, the design of ex-post incentive compatible auctions is a non-trivial problem. This problem becomes even harder when the bidders' values can differ. To enable the design of efficient incentive compatible mechanisms, prior work has introduced useful restrictions on the structure of these valuation functions, such as \emph{submodularity over signals} (SOS) \cite{eden2019combinatorial,SOSimproved}, or constraints across pairs of valuation functions, such as the \emph{single-crossing} property~\cite{milgrom1982theory,maskin1992auctions}.

In this paper, we consider a variety of settings with interdependent values that are not captured by (approximate) SOS or the single-crossing property. We let $k$ be the number of possible values that a bidder's signal can have, and we partition the bidders into $\nExperts$ expertise-groups, depending on the type of information that their signals provide regarding the good being sold. Using these parameters, we prove upper and lower bounds, parameterized by $k$ and $\nExperts$, on the extent to which auctions can approximate the optimal welfare or revenue. All our proposed auctions are ex-post incentive compatible, but our main results also satisfy stronger incentive guarantees: they can be implemented not only as direct-revelation mechanisms (sealed-bid auctions), but also as ascending clock auctions, and they admit the desired outcomes as \emph{obvious ex-post equilibria}~\cite{li2017obviously} which are easy for the bidders to verify, thus leading to more practical solutions.

\subsection{Our Results}
We begin, in Section~\ref{sec:binary}, by considering the interesting case where each bidder's signal regarding the quality of the good can take two possible values, either ``low'' or ``high'', and each bidder's value is a weakly increasing function of these signals. If the valuation function of each bidder is symmetric, i.e., every bidder's signal matters the same, then we provide a clock auction that achieves a $5$-approximation of the optimal social welfare, and a variation of that auction that guarantees revenue that is a $10$-approximation of the optimal social welfare. 
We then generalize these results to non-symmetric functions, where the bidders are partitioned into $\nExperts$ groups based on their expertise, and signals from different groups may have different impact on the values. Our generalization achieves a $5\nExperts$-approximation for social welfare and a $10\nExperts$-approximation for revenue.

In Section~\ref{sec:linear}, we go beyond the case of binary signals and consider problem instances with $k$ distinct signal value options, $\{0, 1, \dots, k-1\}$, allowing for the bidders' signals regarding the quality of the good to be more refined.
The valuation of each bidder can be an arbitrary weakly increasing function of the average quality estimate of each group.
Using a reduction to the binary case, we design a clock auction that achieves a $5\nExperts(k-1)$-approximation for social welfare and a $10\nExperts(k-1)$-approximation for revenue. To complement these positive results, we also prove  a lower bound of $\nExperts(k-1)+1$ for the welfare approximation ratio of ex-post incentive compatible auctions.

Our auctions in these two sections achieve signal discovery using random sampling, while minimizing the probability of rejecting the highest value bidder.
Unlike prior work, our random sampling process is adaptive, depending on prior signal discovery. 
Thus, our auction gradually refines our estimate of the item's quality as perceived by the bidders and eventually decides who to allocate to, aiming to achieve high welfare and revenue. 
Apart from matching the lower bound up to small constants, these auctions crucially also guarantee improved incentives: they admit the desired outcome not just an ex-post equilibrium, but as an \emph{obvious} ex-post equilibrium, making our upper bounds stronger.

Finally, in Section~\ref{sec:general} we consider the most general setting with any number of signals $k>2$  and arbitrary quality functions per expert type. We first prove a stronger lower bound of $\nExperts \binom{k}{2}+1$ for the welfare approximation of ex-post incentive compatible auctions. Then we prove the existence of a universally incentive compatible and individually rational auction that matches this bound.  

Due to space constraints, the proofs of some theorems (particularly those which are similar to previous proofs) have been deferred to the appendix.

\subsection{Related Work}

In an interdependent values setting, a bidder's value for a good may depend on how much others value it. This idea is formally captured by the canonical interdependent values model given by Milgrom and Weber \cite{milgrom1982theory}.  The interdependent values setting has been well-studied in the economics literature for its descriptive ability to capture many real-world scenarios.  Noted examples in the literature include the mineral rights \cite{wilson1969competitive} and common value (e.g., ``wallet game'') models \cite{klemperer1998auctions} discussed above, and the resale model \cite{myerson1981optimal} in which the value a bidder has for a good (e.g., a painting) depends on her own value for the good and the amounts others may be willing to pay on its resale.

A common assumption when studying the interdependent values setting in both the computer science and economics literature is that the valuations of the bidders satisfy  a \emph{single-crossing} condition. Following the definition of Roughgarden and Talgam-Cohen \cite{roughgarden2016optimal}, a set of valuation functions satisfies single-crossing if for all bidders $i$ and $j$ \[\frac{\partial v_i(s_i, \sVec_{-i})}{\partial s_i} \geq \frac{\partial v_j(s_i, \sVec_{-i})}{\partial s_i}.\]  Loosely speaking, single-crossing states that a bidder is more sensitive to her own signal than anyone else is.  Using this assumption, many strong results can be obtained for both welfare and revenue.  For example, Dasgupta and Maskin \cite{dasgupta2000efficient} demonstrated that the celebrated Vickrey-Clarke-Groves (VCG) mechanism can be adapted and extended into the common value setting to obtain optimal welfare given single-crossing.  Ausubel \cite{ausubel1999generalized} demonstrated that a generalized Vickrey auction can achieve efficiency in a multi-unit setting with single-crossing valuations.  For revenue, Li \cite{li2017approximation} and Roughgarden and Talgam-Cohen \cite{roughgarden2016optimal} gave, independently, auctions extracting near optimal revenue in the interdependent values model for any matroid feasibility constraint when the valuations satisfy single-crossing and the signals are drawn from distributions with a regularity-type condition.
Chawla et al. \cite{chawla2014approximate} gave an alternative generalization of the VCG auction with reserve prices and random admission which approximates the optimal revenue in any matroid setting without conditions on signal distributions.  

On the other hand, it is well-known that without single-crossing, achieving the optimal welfare becomes impossible \cite{dasgupta2000efficient,jehiel2001efficient}.  There have thus been recent efforts to \emph{approximate} the optimal welfare when the single-crossing assumption is relaxed.  Eden et al. \cite{eden2018interdependent} suggested a notion called ``$c$-single-crossing'' wherein each bidder is at most a factor $c$ times less sensitive to changes in her own signal than any other bidder is (exact single-crossing has $c = 1$). They gave a $2c$-approximate randomized mechanism when valuation functions are concave and satisfy $c$-single-crossing.  Eden et al. \cite{eden2019combinatorial} proposed an alternative notion termed ``submodularity over signals'' (SOS) which, loosely speaking, stipulates that a valuation function must be less sensitive to increases in any particular signal when the other signals are high.  The authors then gave a randomized $4$-approximate mechanism for all single-parameter downward-closed settings when valuation functions are SOS; this factor was very recently improved to 2 for the case of binary signals by Amer and Talgam-Cohen~\cite{SOSimproved}.  We note that the valuations studied in this paper satisfy neither $c$-single-crossing nor (approximate) SOS, in general. Our work proposes alternative parameterizations of the valuation functions and it provides another step toward a better understanding of interdependent values beyond the classic, and somewhat restrictive, single-crossing assumption.  

In accordance with some recent work in computer science (e.g., see \cite{eden2018interdependent,eden2019combinatorial}), and unlike much of the existing economics literature, we consider a \emph{prior-free} setting where there is no distributional information regarding the signals of the bidders. Thus, our results are in consistent with ``Wilson's doctrine''~\cite{wilson1987game},  which envisions a mechanism design process that is less reliant on the assumption of common knowledge. Our results are independent of an underlying distribution and do not assume that the auctioneer or the bidders have any information regarding each other's signals.

\section{Preliminaries}

We consider a setting where a set $\agents$ of $n$ bidders is competing to receive a good. Each bidder $i \in \agents$ has a private signal $s_i$ regarding the good being sold, which can take one of $k$ publicly known different values. Her valuation of the good, $v_i(\sVec)$, is a publicly known weakly increasing function of the vector of all the bidders' signals, $\sVec = (s_1, s_2, \dots, s_n)$. In many settings of interest it is natural to assume that this is a \emph{symmetric} function over the signals, e.g., when all the bidders have the same access to information, or the same level of expertise. However, we also consider the case when the signal of some bidders may have a different impact than others'. To capture this case we partition the bidders into $\nExperts > 1$ groups and assume that each group has different types of expertise.  In this case, the valuation functions $v_i(\sVec)$ are symmetric with respect to the signals of bidders with the same type of expertise, but arbitrarily non-symmetric across bidders with different types of expertise. Note that this captures arbitrary monotone valuation functions when $\ell=n$, and it also captures several classes of instances where the valuations of different bidders are not (even approximately) single-crossing or SOS.
We call a bidder \emph{optimal for some signal vector $\sVec$} if $i$ is a highest value bidder for that signal profile, i.e., $i\in\arg\max_{j\in N}\{v_j(\sVec)\}$. We use $\opta{\sVec}$ to refer to an optimal bidder for signal vector $\sVec$, breaking ties arbitrarily but consistently if there are multiple optimal bidders for $\sVec$.

In interdependent value settings, a direct-revelation mechanism receives the bidders' signals as input and outputs a bidder to serve and a vector of prices $\pVec(\sVec)$ which each bidder is charged.  For any bidder $i$, the utility $u_i(\sVec) = \val{i}{\sVec} - p_i(\sVec)$ if $i$ is served and $u_i(\sVec) = - p_i(\sVec)$, otherwise.  A mechanism is \emph{ex-post individually rational} if $u_i(\sVec) \geq 0$ for all $i$, assuming all bidders report their true signals. 
A mechanism is \emph{ex-post incentive compatible} if the utility that bidder $i$ receives by reporting her true signal is at least as high as the utility she would obtain by reporting any other signal, assuming all the other bidders report their true signals, i.e., $u_i(s_i, \sVec_{-i}) \geq u_i(s'_i, \sVec_{-i})$ for all $i, \sVec_{-i}$. 
If a mechanism uses randomization, we say that it is \emph{universally ex-post individually rational and ex-post incentive compatible} (universally IC-IR) if it is a distribution over deterministic ex-post individually rational and ex-post incentive compatible mechanisms.

We look to design universally IC-IR randomized mechanisms that aim to serve the bidder with highest realized value given the signal profile.  We measure the expected performance of these mechanisms against the optimal solution given full information.  Given some instance $I$, let $\mechanism(I)$ denote the bidder served by auction $\mechanism$. We then say that $\mechanism$ achieves an $\alpha$-approximation to the optimal welfare for a family of instances $\mathcal{I}$ if
\[\sup_{I \in \mathcal{I}}\frac{\text{max}_{i \in \agents}\{\val{i}{\sVec}\}}{\mathbb{E}\left[\val{\mechanism(I)}{\sVec}\right]} \leq \alpha\]
where the expectation is taken over the random coin flips of our mechanism. In terms of revenue, note that for mechanisms that are individually rational (like the ones that we propose in this paper), we know that the 
revenue of these mechanisms is always a lower bound for their social welfare. We therefore use the optimal social welfare as an upper bound for the optimal revenue and say that $\mechanism$ achieves an $\alpha$-approximation of revenue for a family of instances $\mathcal{I}$ if
\[\sup_{I \in \mathcal{I}}\frac{\text{max}_{i \in \agents}\{\val{i}{\sVec}\}}{\mathbb{E}\left[p_{\mechanism(I)}(\sVec)\right]} \leq \alpha.\]

Our main results in this paper take the form of \emph{clock auctions over signals}.  A clock auction over signals is a multi-round dynamic mechanism in which bidders are faced with personalized ascending signal clocks.  Throughout the auction, the clocks are non-decreasing and, at any point in the auction, a bidder may choose to permanently exit the auction (thereby losing the good permanently).  When a bidder is declared the winner, she is offered a price (greater than or) equal to the value implied by the final clock signals for all bidders.  In a clock auction, a bidder exits the auction if and only if her signal clock is greater than her true signal, we refer to this as \emph{\clockStrategy}. In particular, we seek to design clock auctions where \clockStrategy \ is an \emph{obvious ex-post equilibrium} (OXP) strategy profile \cite{li2017oxp}.  A strategy profile is an OXP of an auction if for any bidder $i$, holding all other bidders' strategies fixed (and assuming they are acting truthfully), the best utility $i$ can obtain by deviating from her truthful strategy under any possible type profile of the other bidders consistent with the history (i.e., their clock signals) is worse than the worst utility $i$ can obtain by following her truthful strategy under any possible  type profile of the other bidders consistent with the history.

\section{Instances with Binary Signal Values}\label{sec:binary}
In this section, we consider the natural case where the signal of each bidder regarding the good can take one of two possible values, e.g., ``low quality'' and ``high quality''. We first focus on instances where the bidders' valuation functions are symmetric over the signals, and we provide a clock auction which admits an ex-post obvious equilibrium and $5$-approximation to the optimal social welfare. We then extend this result to general valuation functions, achieving a $5\nExperts$-approximation to the optimal social welfare. This auction is then also used as a building block for the results of the next section, which considers a setting with $k>2$ signal values.

\subsection{A Clock Auction for Symmetric Valuation Functions}

A central result of this paper is the \emph{signal discovery auction}, which is presented as a sealed-bid auction below (see Mechanism \ref{alg:sampling}), but can also be implemented as a clock auction (see Theorem~\ref{lem:clockAuctionBinary}). This auction aims to discover how many bidders have a high signal, while minimizing the probability that the optimal bidder is rejected during the discovery process. Throughout the execution of the auction, the set $A$ includes the bidders that remain active, i.e., the ones that have not been rejected yet. The variables $q_{min}$ and $q_{max}$ provide a lower and an upper bound, respectively, for the number of bidders that have a high signal, based on the signals discovered up to that point. Note that $q_{min}$ is initialized to $0$ and $q_{max}$ is initialized to $n$, corresponding to all bidders having signal $0$ or signal $1$, respectively. The set $R^*$ contains all the bidders that have been rejected, without first verifying that they are not optimal. 

The auction uses randomized sampling in order to initiate this discovery process: it chooses one of the active bidders uniformly at random, it rejects that bidder, and then uses its signal value to narrow down the range $[q_{min}, q_{max}]$. We refer to this as a ``costly'' signal discovery, because it may lead to the rejection of the highest value bidder. Then, this discovery leads to a sequence of ``free'' signal discoveries, by using this information to identify active bidders that cannot be optimal, rejecting them, and then using their signal to further narrow down the $[q_{min}, q_{max}]$ range.
When no additional free signal discoveries are available, the auction removes any bidder of $R^*$ that is now verified to be non-optimal, and executes another costly signal discovery.

This process continues until there is only one active bidder, at which point this bidder is declared the winner. We say that a signal profile $\sVec$ is consistent with some $q \in [q_{min}, q_{max}]$ if it contains a number of ``high'' signals equal to $q$.  If this bidder $i$ is optimal for a signal profile $\sVec$ consistent with exactly one $q\in [q_{min}, q_{max}]$, then the bidder is awarded the good at price $p=v_i(\sVec)$; if the bidder is optimal for multiple signal profiles consistent with distinct numbers of ``high'' signal bidders in $[q_{min}, q_{max}]$, she is awarded the good at the price corresponding to a signal profile with the fewest number of ``high'' signal bidders. 

\begin{algorithm}[h]
\label{alg:sampling}
\SetAlgoLined
\DontPrintSemicolon
\SetAlgoNoEnd
 Let $A\leftarrow \agents$, $R^*\leftarrow \emptyset$, $q_{min} \leftarrow 0$, and $q_{max} \leftarrow n$\;
 \While{$|A| > 1$}{
    \tcp{A ``costly'' signal discovery}
    Select a bidder $i\in A$ uniformly at random\;
    Let $A\leftarrow A \setminus \{i\}$ and $R^*\leftarrow R^* \cup \{i\}$\;
    \eIf{$s_i = \lowsig$}{
        $q_{max} \leftarrow q_{max}-1$\;}
        { 
        $q_{min} \leftarrow q_{min}+1$\;}
    \tcp{A sequence of ``free'' signal discoveries}
    \While{$\exists j\in A$ that is not optimal for any $\sVec$ consistent with some $q\in [q_{min}, q_{max}]$}
    {$A \leftarrow A\setminus\{j\}$\;
    \eIf{$s_j = \lowsig$}{
        $q_{max} \leftarrow q_{max}-1$\;}
        { 
        $q_{min} \leftarrow q_{min}+1$\;}}
    \While{$\exists j\in R^*$ that is not optimal for any $\sVec$ consistent with some $q\in [q_{min}, q_{max}]$}
    {$R^* \leftarrow R^*\setminus\{j\}$\;}
}
 Let $i$ be the single bidder in $A$\;
 Let $s'_i$ be the smallest signal such that $i$ is optimal for $(s'_i, \sVec_{-i})$\; \label{lst:line:Price}
 \If{$\val{i}{\sVec} \geq \val{i}{(s'_i, \sVec_{-i})}$} 
 {Award the good to $i$ at price $\val{i}{(s'_i, \sVec_{-i})}$\; \label{lst:line:Alloc}} 
\caption{Signal discovery auction for binary signal values}
\end{algorithm}

The following lemma shows that the size of $R^*$ is never more than 2, which allows us to bound the probability that the auction identifies the optimal bidder.

\begin{lemma}\label{lem:Rbound}
Throughout the execution of the signal discovery auction, the size of $R^*$ is never more than 2. 
\end{lemma}
\begin{proof}

We first note that, throughout the auction, the only bidders in $A \cup R^*$ are the potentially optimal bidders (i.e., those which correspond to some possible signal profile) since bidders are removed from $A \cup R^*$ when they are determined to be non-optimal.
Initially $R^*$ is empty and at the beginning of each iteration of the outer while-loop, one randomly sampled active bidder $i$ is added to this set, increasing its size by one. The signal of bidder $i$ is then used to update either $q_{min}$ or $q_{max}$; if $s_i=0$ the auction can infer that $q_{max}$ is not the true number of high signal bidders, and if $s_i=1$ the auction can infer that $q_{min}$ is not the true number of high signal bidders. In both of these cases, some possible symmetric signal profile is ruled out, and this may lead to a sequence of ``free'' signal discoveries, as discussed below.

Whenever a symmetric signal profile $\sVec$ is ruled out, there are four possibilities regarding the bidder who is optimal for that level, i.e., the bidder $\opta{\sVec}$:
\begin{enumerate}
    \item If $\opta{\sVec}$ is in $A$ and is not optimal for any other $\sVec'$ consistent with some number $q$ of high signal bidders in the updated interval $[q_{min}, q_{max}]$, then the first inner-while loop of the auction will remove that bidder from $A$ and use its signal to rule out one more quality level.
    \item If $\opta{\sVec}$ is in $A$ and is also optimal for some other $\sVec'$ consistent with some number $q$ of high signal bidders in the updated interval $[q_{min}, q_{max}]$, then the iteration of the outer while-loop terminates without any additional operations and we proceed to the next iteration.
    \item If $\opta{\sVec}$ is in $R^*$, and is not optimal for any other $\sVec'$ consistent with some number $q$ of high signal bidders in the updated interval $[q_{min}, q_{max}]$, then the second inner while-loop removes $\opta{\sVec}$ from $R^*$ and we proceed to the next iteration of the outer while-loop.
    \item If $\opta{\sVec}$ is in $R^*$, and is also optimal for some other $\sVec'$ consistent with some number $q$ of high signal bidders in the updated interval $[q_{min}, q_{max}]$, then the iteration of the outer while-loop terminates without any additional operations and we proceed to the next iteration.
\end{enumerate}

Considering these four possibilities, note that while the first case arises, the execution remains in the first inner while-loop and the size of $R^*$ remains unchanged. When the third case arises, the size of $R^*$ is first reduced by one (because the auction enters the second inner while-loop) and then proceeds to the next iteration of the outer while-loop, which may bring this up to the same size again. As a result, the third case does not increase the size of $R^*$ either. 

On the other hand, both cases 2 and 4 may lead to an increase in the size of $R^*$ by 1, since they terminate the current iteration of the outer while-loop and may proceed to the next one, which would add one more bidder to $R^*$. 

However, at the end of each iteration of the outer while-loop, $A$ and $R^*$ contain only bidders that are optimal for some $\sVec$ consistent with some number of high signal bidders $q$ in $[q_{min}, q_{max}]$ (all the others are removed from $A$ in the first inner while-loop and from $R^*$ in the second inner while-loop). Also, at the end of each iteration of the outer while-loop, we have $q_{max}=q_{min}+|A|$. To verify this fact note that the signal of everyone not in $A$ has already been used to update the interval $[q_{min}, q_{max}]$ and the only signals not used yet are those of the bidders in $A$. If all the bidders in $A$ have a low signal, then the true $\sVec$ has $q_{min}$ bidders with high signals. If they all have a high signal (adding $|A|$ bidders with high signal), the true $\sVec$ has $q_{max}$ bidders with high signals.

Therefore, we know that at the end of each iteration of the outer while-loop, every bidder in $A$ and $R^*$ is optimal for some possible symmetric signal profile with a number of high value bidders in $[q_{min}, q_{max}]$ and there are at most $|A|+1$ such distinct signal profiles. If $R^*$ is empty at that point, this means that there can be at most one bidder in $A$ that is optimal for two distinct signal profiles. If $|R^*|=1$, then there are $|A|+1$ optimal bidders and $|A|+1$ distinct signal profiles, so there is no bidder in $A$ or $R^*$ that is optimal for more than one such profile. This means that in the next iteration of the outer while-loop, cases 2 and 4 listed above cannot arise, and therefore the size of $R^*$ cannot be strictly more than 1 at the end of any iteration of the outer while loop.
\qed
\end{proof}

\begin{theorem}\label{thm:baseTheorem}
The signal discovery auction achieves a $5$-approximation of the optimal welfare for instances with binary signals.
\end{theorem}
\begin{proof}
Let $i^*$ be the optimal bidder and $q^*$ be the true number of high signals.
We first observe that a bidder is removed from $A \cup R^*$ only if they are determined to be non-optimal.  Thus, we know that $i^* \in A \cup R^*$ throughout the running of the algorithm.  By Lemma \ref{lem:Rbound} we know that $|R^*| \leq 2$ throughout the running of the algorithm. There are then at most $5$ distinct bidders who can be in $A \cup R^*$ at the end of the algorithm: $i^*$ and the (up to) four other bidders optimal for signal profiles corresponding to $q^*-2$, $q^*-1$, $q^*+1$, or $q^*+2$ high signal bidders.  Provided that these four other bidders enter $R^*$ (or are eliminated) before $i^*$ is added to $R^*$ we then obtain the optimal welfare.   We conclude by noting that, since the choices of the bidder to be added to $R^*$ is made uniformly at random, we can envision the order in which bidders are added to $R^*$ as a uniform at random permutation over the bidders fixed at the outset.  In a uniform random permutation, $i^*$ follows these four bidders with probability $1/5$. \qed

\end{proof}

The signal discovery auction, as presented, achieves no interesting worst-case approximation for revenue when the benchmark is the ex-post optimal welfare. In particular, if there is a single optimal bidder for all the signal profiles corresponding to numbers of high signal bidders in $[q_{min}, q_{max}]$, and the true number of high signal bidders is $q_{max}$, the mechanism charges the winner $i$ a price of $v_i(\sVec')$ where $\sVec'$ is the signal profile obtained by her signal being $0$ (corresponding to $q_{min}$ high signal bidders). If the true signal profile $\sVec''$ corresponds to having $q_{max}$ bidders of high signal, the ex-post optimal welfare is $v_i(\sVec'')$, which can be arbitrarily higher than $v_i(\sVec')$. To address this issue, our next result shows that if we slightly modify the pricing rule of the mechanism, then we can achieve revenue which is a $10$-approximation of the ex-post optimal welfare (which simultaneously also implies that the welfare we obtain is a $10$-approximation).

\begin{theorem}\label{thm:revTheorem}
The pricing rule of the signal discovery auction can be adjusted to achieve revenue which is a $10$-approximation of the optimal welfare for instances with binary signals.
\end{theorem}
\begin{proof}
If in line \ref{lst:line:Price} of Mechanism \ref{alg:sampling} we instead select a $\sVec$ for which $i$ is optimal consistent with some random $q' \in [q_{min}, q_{max}]$ and $\sVec$ is the true signal profile, we extract all of the welfare as revenue.  Since $i$ is the only bidder with unknown signal value, there are at most two levels for which $i$ is optimal so we select the signal profile with probability $1/2$, yielding the $10$-approximation. Note that in line \ref{lst:line:Alloc} we only allocate the item if the price is below the true value of $i$, so we preserve ex-post IC-IR with this modification. \qed
\end{proof}

We conclude this section by verifying that the outcome of the signal discovery auction can be implemented as an obvious ex-post equilibrium~\cite{li2017oxp}.

\begin{theorem}\label{lem:clockAuctionBinary}
The signal discovery auction can be implemented as an ascending clock auction over the signals wherein {\clockStrategy} is an obvious ex-post equilibrium. 
\end{theorem}
\begin{proof}
Rather than asking bidders to report their signals we may instead equip each bidder with a signal clock.  The clocks of all bidders begin at $0$ and when bidder $i$ would have her signal discovered by the above mechanism, we instead raise the clock of $i$ to $1$.  If $i$ rejects the new clock signal level (i.e., permanently exits the auction), she cannot win the item \emph{regardless of her beliefs} about the signals of the remaining bidders.  

If the true signal of $i$ is $1$, for any profile of signals of the remaining bidders (assuming these signals are true) the worst utility $i$ can obtain by accepting the increased clock signal level is $0$ (by losing the item or by winning the item and being charged exactly her welfare).  Thus, at any point in the auction, regardless of the history, when $i$ is approached to increase her clock signal level, the best utility $i$ can obtain by not accepting the increased clock signal level (thereby necessarily losing the good) is weakly less than the worst utility $i$ can obtain by accepting the increased clock signal level.  On the other hand, if the true signal of $i$ is $0$, for any profile of signals of the remaining bidders (assuming these signals are true) if she instead accepts the increased clock signal level she either will continue to lose the auction (thereby obtaining a utility of $0$) or win the auction at a quality level higher than the actual underlying quality of the good.  Since the threshold signal of $i$ would then be $1$, she would necessarily be charged a price weakly higher than her value for the good and she would obtain non-positive utility.  Thus, in either case, truthfully responding whether or not the clock signal level is above a bidder's signal is an obvious ex-post equilibrium. \qed

\end{proof}

\begin{corollary}
The version of the signal discovery auction which obtains revenue guarantees can also be implemented as an ascending clock auction over the signals wherein {\clockStrategy} is an obvious ex-post equilibrium.
\end{corollary}
\begin{proof}
The proof follows exactly as above except we raise the clock signal level of the winning bidder to the one corresponding to the randomly selected signal profile (effectively setting a take-it-or-leave-it price at this signal). \qed
\end{proof}

\subsection{A Clock Auction for General Valuation Functions}

In this section, we demonstrate how our auction for symmetric valuation functions, i.e., the case where $\ell=1$, above, can be easily extended to handle general valuation functions over binary signals, leading to approximation bounds that depend on the number of expert-groups, $\nExperts$.  

The mechanism first uniformly at random selects some $\nExperts'\in \{1, 2, \dots, \nExperts\}$, and then assumes that the optimal bidder belongs to expertise type $\nExperts'$.  The mechanism rejects all bidders outside expertise-group $\nExperts'$ and ``learns'' their signals. 
The auction then knows all the signals of bidders not in $\nExperts'$ and the problem reduces to also discovering the number of bidders in $\nExperts'$ that have a high signal. We can therefore run Mechanism \ref{alg:sampling} among the bidders in $\nExperts'$ to decide the winner among them, and the price offered to her.

\begin{theorem}
The above mechanism yields a $5\nExperts$-approximation of the optimal welfare for instances with binary signals.
\end{theorem}
\begin{proof}
The probability that the optimal bidder does, indeed, belong to the expertise-group $\nExperts'$ is $1/\nExperts$. If the mechanism guesses the value of $\nExperts'$ correctly, then the rejection of all the other bidders comes at no cost, and it reduces the problem to finding the optimal bidder within the group $\nExperts'$. But, since we now know all the signal values of bidders outside the group $\nExperts'$, we can use Mechanism \ref{alg:sampling} to discover the optimal bidder with probability at least $1/5$ (by Theorem \ref{thm:baseTheorem}). Combining these observations, the above mechanism allocates to the optimal bidder with probability at least $1/(5\nExperts)$.
\qed

\end{proof}

\begin{theorem}\label{thm:revBinaryExperts}
The pricing rule of the above mechanism can be adjusted to achieve revenue which is a $10\nExperts$-approximation of the optimal welfare for instances with binary signals.
\end{theorem}
\begin{proof}
Similar to above, the proof of this theorem closely follows the proof of Theorem \ref{thm:revTheorem}.  Given that we select the type of expertise with the welfare-optimal bidder we obtain revenue which is a $10$-approximation to the welfare by running Mechanism \ref{alg:sampling} but instead selecting a random quality $q' \in [q_{min}, q_{max}]$ for which $i$ is optimal at line \ref{lst:line:Alloc}. \qed
\end{proof}

\section{Shared Quality Functions over $k$ Signal Values}\label{sec:linear}
We now move beyond instances with binary signals and consider a class of valuation functions over $k\geq 2$ signal values. Each bidder $i$'s signal can take any value $s_i\in \{0, 1, \dots, k-1\}$ and the average of these signals determines the \emph{quality of the good} $q=\sum_{i\in N}s_i$ (note that the average of the signals can be directly inferred from the sum, so we use the sum for simplicity of notation).\footnote{Note that the actual $k$ signal values need not be $\{0,1,\dots, k-1\}$, but we need them to be equidistant for our results to hold.} This captures a variety of settings where each bidder has some estimate regarding the quality, but the true quality is best approximated by averaging over all the bidders' signals (e.g., see the \emph{wisdom of the crowds} phenomenon \cite{surowiecki2005wisdom}). Each bidder $i$'s value for the good is provided by some (arbitrary) weakly increasing function $v_i(q)$, which depends only on $q$, quantifying how much each bidder values quality.

Apart from these symmetric valuation functions, we also consider non-symmetric ones involving $\nExperts$ different classes of experts. The bidders are partitioned into sets $N_1, N_2, \dots, N_{\nExperts}$, depending on their expertise, and the quality estimate from each expert group $\nExperts'$ is their average signal, i.e., $q_{\nExperts'}=\sum_{i\in N_{\nExperts'}}s_i$. In this case, the quality of the good is captured by the \emph{shared quality vector} $\qVec=(q_1, q_2, \dots, q_{\nExperts})$, and each bidder's valuation is a function $v_i(\qVec)$. The only restriction on the valuation function is that it is weakly increasing with respect to the underlying signals, but it can otherwise arbitrarily depend on the quality vector. For instance, this allows us to model settings where the signals of each group of experts imply the quality of the good with respect to some dimension, and each bidder can then synthesize this information into a quite complicated valuation function, depending on the aspects that she cares about the most.
 
In this section, we first provide a lower bound for the approximability of the optimal social welfare by universally ex-post IC-IR auctions, parameterized by $k$ and $\nExperts$. We then provide a way to leverage the ideas from the previous section to achieve essentially matching upper bounds using clock auctions and ensuring incentive guarantees even better than ex-post IC-IR.

\subsection{Approximation Lower Bound for ex-post IC-IR auctions}
We first prove a lower bound for the welfare approximation that one can achieve for the class of instances of this section involving $\nExperts$ types of experts with $k$ signal values each. It is worth noting that the construction for this lower bound is based on a simple class of valuation functions that only depend on the weighted average of the bidder's signals (with each expert group having a different weight coefficient). Also, for the case $k=2$, i.e., the binary case considered in the previous section, this implies a lower bound of $\nExperts + 1$.

\begin{theorem}
No ex-post incentive compatible auction with $\nExperts$ types of experts and shared quality functions can achieve better than an $\nExperts (k-1) + 1$-approximation to the optimal welfare.
\end{theorem}
\begin{proof}
We consider a particularly simple setting, in which the quality of the good can be summarized as a weighted average of all the bidders' signals (with bidders from different expertise classes given different weights).  Note that this is readily captured by the model described above.  
It follows that when we reduce the signal of $i$ by $d>0$, the quality of the good changes by $d w_i$. Note that $d$ can be at most $k-1$ different values. We construct a valuation function as follows.  For each $j \in \{0,1,\dots,\ell-1\}$, we define the valuation function of bidder $i$ where $(k-1)\cdot j + 1 \leq i \leq (k-1) \cdot (j+1)$ as follows:
\begin{equation*}
    \val{i}{t} =
    \begin{cases}
        \Delta_i & \text{if }~ t \geq S - (i - (k-1)\cdot j) \cdot w_i, \\
        0 & otherwise.
    \end{cases}
\end{equation*}
Finally, for bidder $i' = (k-1)\cdot\ell + 1$ (who has signal $0$ in $\sVec$), $\val{i'}{t} = \Delta_{i'}$ when $t \geq S$ and $\val{i'}{t} = 0$ otherwise.  
We let $\Delta_1 = 1$ and $\forall i > 1$, $\Delta_{i} = H\Delta_{i-1}$ ($H$ is arbitrarily large).  In other words, at any of these qualities, we must allocate the good to the optimal bidder with probability $1/\alpha$ to obtain an $\alpha$-approximation to the optimal welfare in the worst case.  To obtain a $\nExperts (k-1) + 1 - \epsilon$ approximation for $\epsilon > 0$, it then must be that we allocate the good to the optimal bidder at all of these qualities with probability at least $1/(\nExperts (k-1) + 1 - \epsilon)$.  But then we have that for all $d \in \{1, 2, \dots, k-1\}$ and $w \in \{1, k,\dots, k^{\nExperts - 1}\}$ if we allocate the good to the optimal bidder $i$ when the quality is $S - dw$ with probability $p$, we must continue to allocate the good to $i$ with probability $p$ when the quality is $S$ in order to maintain universal ex-post incentive compatibility (by monotonicity of an allocation rule).  Finally, since there are $\nExperts (k-1) + 1$ qualities identified above, each of the distinct optimal bidders at these qualities must be allocated the good with probability at least $1/(\nExperts (k-1) + 1 - \epsilon)$ at quality $S$, a contradiction.
\qed
\end{proof}

\subsection{A Clock Auction for Instances with Shared Quality Functions}
We now provide a way to reduce this problem to the case of binary signals, while losing only a $k-1$ factor in our bounds. As a result, the induced upper bounds closely approximate the lower bound provided above. The majority of this section discusses how Mechanism~\ref{alg:ksampling} achieves this reduction for the case where $\nExperts$, and then briefly explain how to generalize our bounds for instances with $\nExperts>1$.

Similarly to Mechanism~\ref{alg:sampling} in the binary setting, whose goal is to discover the number of signals that are high, Mechanism~\ref{alg:ksampling} aims to discover the value of the sum of the signals. 
Throughout its execution, the auction maintains an interval $[q_{min}, q_{max}]$ such that the true sum $q$ is guaranteed to be in that interval.  It gradually refines this range by discovering bidder signals as in the binary setting. The main difference is that we now need to be more careful in order to ensure that the size of $R^*$ remains low. To achieve this, the auction chooses some $m \in \{0, 1, \dots, k-2\}$ uniformly at random and assumes that $q\mod (k-1)=m$. It thus randomly reduces the number of values of $q$ that it considers from $n(k-1)+1$ (since the sum can initially range from 0 to $n(k-1)$) to just $n+1$ (which is equal to the length of the $[q_{min}, q_{max}]$ interval in the case of binary signals). Importantly, the values of $q$ that are considered after this sampling are spaced apart by $k-1$, allowing us to upper bound the size of $R^*$.

\begin{algorithm}[h]
\label{alg:ksampling}
\SetAlgoLined
\DontPrintSemicolon
\SetAlgoNoEnd
 Let $A\leftarrow \agents$, $R^*\leftarrow \emptyset$, $q_{min} \leftarrow 0$, and $q_{max} \leftarrow n(k-1)$\;
 Choose some $m \in \{0, 1, \dots, k-2\}$ uniformly at random\;
 $S \leftarrow \{q\in [q_{min}, q_{max}]~|~ q\mod (k-1) = m \}$ \;
 \While{$|A| > 1$}{
    \tcp{A ``costly'' signal discovery}
    Select a bidder $i\in A$ uniformly at random\;
    Let $A\leftarrow A \setminus \{i\}$ and $R^*\leftarrow R^* \cup \{i\}$\;
    $q_{max} \leftarrow q_{max}-(k-1-s_i)$\\
      $q_{min} \leftarrow q_{min} +s_i$ \\
    \tcp{A sequence of ``free'' signal discoveries}
    \While{$\exists j\in A$ that is not optimal for any $q\in S\cap [q_{min}, q_{max}]$}
    {$A \leftarrow A\setminus\{j\}$\;
    $q_{max} \leftarrow q_{max} +s_j-k+1$\\
      $q_{min} \leftarrow q_{min} +s_j$\\
    }
    \While{$\exists j\in R^*$ that is not optimal for any $q\in S\cap  [q_{min}, q_{max}]$}
    {$R^* \leftarrow R^*\setminus\{j\}$\;}
}
 Let $i$ be the single bidder in $A$\;
 Choose the smallest quality level $q'\in S\cap [q_{min}, q_{max}]$ for which $i$ is optimal\; \label{lst:kprice}
 \If{$\val{i}{\quality{\sVec}} \geq \val{i}{q'}$} 
 {Award the good to $i$ at price $\val{i}{q'}$\;} \label{lst:kpriceIR}
\caption{Signal discovery auction for $k$ signal values}
\end{algorithm}

\begin{lemma}\label{lem:Rboundk}
The  set of $R^*$ in Mechanism~\ref{alg:ksampling} is never more than 2.

\end{lemma}
\begin{proof}
We first note that, throughout the auction, the only bidders in $A \cup R^*$ are the potentially optimal bidders (i.e., those which correspond to some possible signal profile) since bidders are removed from $A \cup R^*$ when they are determined to be non-optimal. Initially $R^*$ is empty and at the beginning of each iteration of the outer while-loop, one randomly sampled active bidder $i$ is added to this set, increasing its size by one. The signal of bidder $i$ is then used to update $q_{min}$ and $q_{max}$ possibly eliminating a quality from $S \cap [q_{min}, q_{max}]$ and this may lead to a sequence of ``free'' signal discoveries, as discussed below.

Whenever a level $q$ is ruled out, there are four possibilities regarding the bidder who is optimal for that level, i.e., the bidder $\opta{q}$:
\begin{enumerate}
    \item If $\opta{q}$ is in $A$ and is not optimal for any other quality level in the updated interval $S \cap [q_{min}, q_{max}]$, then the first inner-while loop of the auction will remove that bidder from $A$ and use its signal to rule out one more quality level.
    \item If $\opta{q}$ is in $A$ and is also optimal for some other quality level in the updated interval $S \cap [q_{min}, q_{max}]$, then the iteration of the outer while-loop terminates without any additional operations and we proceed to the next iteration.
    \item If $\opta{q}$ is in $R^*$, and is not optimal for any other quality level in the updated interval $S \cap [q_{min}, q_{max}]$, then the second inner while-loop removes $\opta{q}$ from $R^*$ and we proceed to the next iteration of the outer while-loop.
    \item If $\opta{q}$ is in $R^*$, and is also optimal for some other quality level in the updated interval $S \cap [q_{min}, q_{max}]$, then the iteration of the outer while-loop terminates without any additional operations and we proceed to the next iteration.
\end{enumerate}

Considering these four possibilities, note that while the first case arises, the execution remains in the first inner while-loop and the size of $R^*$ remains unchanged. When the third case arises, the size of $R^*$ is first reduced by one (because the auction enters the second inner while-loop) and then proceeds to the next iteration of the outer while-loop, which may bring this up to the same size again. As a result, the third case does not increase the size of $R^*$ either. 

On the other hand, both cases 2 and 4 may lead to an increase of the size of $R^*$ by 1, since they terminate the current iteration of the outer while-loop and may proceed to the next one, which would add one more bidder to $R^*$.  However, at the end of each iteration of the outer while-loop, $A$ and $R^*$ contain only bidders that are optimal for some level in $S \cap [q_{min}, q_{max}]$ (all the others are removed from $A$ in the first inner while-loop and from $R^*$ in the second inner while-loop). At the end of each iteration of the outer while-loop, we have $q_{max}=q_{min}+(k-1) \cdot |A|$ since the only signals still not used are those of bidders in $A$ and these signals range from $0$ to $k-1$.  On the other hand, the consecutive qualities in $S$ are $k-1$ apart, so $|S \cap [q_{min}, q_{max}]| \leq |A| + 1$ (with equality if both $q_{min}$ and $q_{max}$ are in $S$).  

Therefore, we know that at the end of each iteration of the outer while-loop, every bidder in $A$ and $R^*$ is optimal for some level in $S \cap [q_{min}, q_{max}]$ and there are at most $|A|+1$ levels in that interval. If $R^*$ is empty at that point, this means that there can be at most one bidder in $A$ that is optimal for two levels in the interval. If $|R^*|=1$, then there are $|A|+1$ optimal bidders and $|A|+1$ optimal levels, so there is no bidder in $A$ or $R^*$ that is optimal in more than one levels. This means that in the next iteration of the outer while-loop, cases 2 and 4 given in Lemma \ref{lem:Rbound} cannot arise, and therefore the size of $R^*$ cannot increase beyond 1 by the end of any iteration of the outer while loop. \qed
\end{proof}

\begin{theorem} \label{thm:additiveKWelfare}
The signal discovery auction achieves a $5(k-1)$ approximation of the optimal welfare for instances with shared quality functions.
\end{theorem}
\begin{proof}
Given that we guess the correct quality residue (i.e., the optimal quality is in the set $S$ we randomly select) the proof follows Theorem \ref{thm:baseTheorem} and its proof directly.  Let $q^*$ denote the true quality and $i^*$ the optimal bidder for this quality.  
We first observe that a bidder is removed from $A \cup R^*$ only if they are determined to be non-optimal.  Thus, we know that $i^* \in A \cup R^*$ throughout the running of the algorithm.  By Lemma \ref{lem:Rbound} we know that $|R^*| \leq 2$ throughout the running of the algorithm.  Thus, there are at most $5$ distinct bidders who can be in $A \cup R^*$ at the end of the algorithm -- $i^*$ and the (up to) four other bidders optimal for signal profiles corresponding to $q^*-2$, $q^*-1$, $q^*+1$, or $q^*+2$ high signal bidders.  Provided that these four other bidders enter $R^*$ before $i^*$ we then obtain the optimal welfare.  We conclude by noting that, since the choices of the bidder to be added to $R^*$ is made uniformly at random, we can envision the order in which bidders are added to $R^*$ as a uniform at random permutation over the bidders fixed at the outset.  In a uniform random permutation, $i^*$ follows these four bidders with probability $1/5$, so we obtain a $5p$-approximation where $p$ is the probability with which we select the correct residue.  Since we guess this uniformly at random, we obtain a $5(k-1)$-approximation.
\qed
\end{proof}

\begin{theorem} \label{thm:additiveKRevenue}
The pricing rule of the signal discovery auction can be adjusted to achieve revenue which is a $10(k-1)$ approximation of the optimal welfare for instances with shared quality functions.
\end{theorem}
\begin{proof}
Given that we guess the correct quality residue which occurs with probability $1/(k-1)$, if in line \ref{lst:kprice} of Mechanism \ref{alg:ksampling} we instead select a random quality $q' \in S \cap [q_{min}, q_{max}]$ for which $i$ is optimal and this is the true quality, we extract all of the welfare as revenue.  Note that $i$ is the only bidder with unknown signal and this signal can only affect the quality by at most $k-1$.  Also, adjacent quality levels in $S \cap [q_{min}, q_{max}]$ are $k-1$ apart.  Thus, there are at most two levels in $S \cap [q_{min}, q_{max}]$ for which $i$ is optimal so we select the correct quality with probability $1/2$, yielding the $10(k-1)$-approximation. Note that in line \ref{lst:kpriceIR} we only allocate the item if the price is below the true value of $i$, so we preserve ex-post IC-IR with this modification. \qed

\end{proof}

\begin{theorem} \label{lem:clockAuctionAdditive}
Mechanism \ref{alg:ksampling} can be implemented as an ascending clock auction over the signals wherein {\clockStrategy} is an obvious ex-post equilibrium.
\end{theorem}
\begin{proof}
The proof follows closely the proof of Theorem \ref{lem:clockAuctionBinary}. Rather than asking bidders to report their signals we may instead equip each bidder with a signal clock.  The clock of all bidders begin at $0$ and when bidder $i$ would be sampled by the above mechanism, we instead raise the clock of $i$ by one signal level at a time until she rejects the signal level or reaches signal $k-1$.\footnote{To ensure that bidders cannot ``identify'' that they have been used for signal discovery, we can, with some small probability, instead run a uniform ascending signal clock auction which awards the item to the bidder with the highest signal.  This sacrifices a small constant in the approximation guarantee, but enhances the practicality of our auction as all bidders are incentivized to continue to remain in the auction even as their clock is rising (which would indicate that they were used for signal discovery in Mechanism \ref{alg:ksampling}).}  If $i$ rejects the new clock signal level (i.e., permanently exits the auction), she cannot win the item (as above) \emph{regardless of her beliefs} about the signals of the remaining bidders.  

If the true signal of $i$ is $k-1$, for any profile of signals of the remaining bidders (assuming these signals are true) the worst utility $i$ can obtain by accepting the clock signal level throughout the raising process (until it reaches $k-1$) is $0$ (by losing the item or by winning the item and being charged exactly her welfare).  Thus, at any point in the auction, regardless of the history, when $i$ is approached to increase her clock signal level, the best utility $i$ can obtain by not accepting all increased clock signal levels (thereby necessarily losing the good) is weakly less than the worst utility $i$ can obtain by accepting the increased clock signal level.  On the other hand, if the true signal of $i$ is strictly less than $k-1$, for any profile of signals of the remaining bidders (assuming these signals are true) if she instead continues to accept the clock signal level even at $k-1$ (if she rejects at any point she loses the item regardless of when she rejects) she either will continue to lose the auction (thereby obtaining a utility of $0$) or win the auction at a quality higher than the actual underlying quality of the good.  Since the threshold signal of $i$ would then be $k-1$, she would necessarily be charged a price weakly higher than her value for the good and she would obtain non-positive utility.  Thus, in either case, truthfully responding whether or not the clock signal level is above a bidder's signal is an obvious ex-post equilibrium. \qed

\end{proof}

\begin{theorem}\label{thm:expertsKAdditive}
Mechanism \ref{alg:ksampling} can be modified to yield a $5(k-1)\nExperts$-approximation of the optimal welfare and achieve revenue which is a $10(k-1)\nExperts$-approximation of the optimal welfare for shared quality functions with $\nExperts$ expertise types.
\end{theorem}
\begin{proof}
The proof of this theorem follows directly from the proofs of Theorems \ref{thm:additiveKWelfare} and \ref{thm:additiveKRevenue}.  

Consider the algorithm which first uniformly at random selects some type $\nExperts'$ of expertise and guesses that the optimal bidder belongs to this type.  It then samples all bidders outside expertise type $\nExperts'$ and updates $q_{\text{min}}$ to be the quality obtained by the signal vector of the sampled bidders' signals and $0$ for all bidders of expertise $\nExperts'$.  Finally, the auction then restricts attention to the qualities which may arise for all possible signal profiles of bidders of expertise $\nExperts'$ and runs Mechanism \ref{alg:ksampling} for these bidders and qualities. 

We begin with the welfare guarantee.  Given that we select the type of expertise with the welfare-optimal bidder we obtain welfare which achieves a $5(k-1)$-approximation to the optimal welfare from running Mechanism \ref{alg:sampling} on the correct type of expertise (discarding the other bidders).  Since we select every expertise type with equal probability (i.e., $1/\nExperts$) we obtain a $5(k-1)\nExperts$-approximation as desired. For revenue, given that we select the type of expertise with the welfare-optimal bidder we obtain revenue achieving a $10(k-1)$-approximation to the optimal welfare by instead selecting a random quality level $q' \in S \cap [q_{min}, q_{max}]$ for which $i$ is optimal (see the proof of Theorem \ref{thm:additiveKRevenue} for additional details).  Since we select the correct expertise type with probability $1/\nExperts$ we obtain the $10(k-1)\nExperts$-approximation as desired.\qed
\end{proof}

\section{General Valuation Functions and Signal Values}\label{sec:general}

We now turn to the more general case where the quality of the good is any weakly increasing function of the signals that treats bidders with the same expertise type symmetrically. We provide an approximation lower bound for 
any allocation function that is \emph{monotone}: a necessary condition of ex-post IC. We conclude our results by proving that there exists a universally IC-IR auction that matches the approximation ratio lower bound. We can adjust the mechanism to achieve revenue that is  $k\cdot (\nExperts\binom{k}{2} +1)$ approximation of the welfare. 
 
 \begin{theorem}\label{thm:genlb}
 No ex-post incentive compatible auction can get more than a $\nExperts \binom{k}{2}+1$-approximation to the optimal welfare. Also, no universally ex-post IC-IR auction can obtain revenue more than a  $\nExperts \binom{k}{2}+1$ fraction of the revenue. 
 \end{theorem}
  \begin{proof}

Similarly to the shared quality functions case, we will construct a set of bidders, valuation functions, and signal profiles such that we can find some signal vector for which at least $\nExperts\binom{k}{2} +1$ bidders all require positive probability of allocation by enforcing monotonicity constraints required to achieve ex-post incentive compatibility.

We define the set of bidders
$$\agents=\{i^*\}\cup \{(i,j,l)\mid i\in\{1,\dots,k-1\}, j\in \{0,\dots,i-1\},l=\{1,\dots,\nExperts\}\}.$$

Note that the total number of bidders is exactly  $|N|=\nExperts\binom{k}{2}+1$. We consider the following signal profile $\sVec$ where $s_{(i,j,l)}=i$ and $s_{i^*}$ is arbitrary. We design a valuation function such that either monotonicity or feasibility would be violated at $\sVec$ if the allocation achieves an approximation ratio better than $|\agents|$.

Additionally, we define the set of signal profiles $S_{(i,j,l)}$ where a bidder with expertise type $l$ and signal $i$ reduces her signal to $j<i$. Due to our constraints, any valuation function has the same outputs among all signal profiles in any $S_{(i,j,l)}$.

Now we are ready to define the valuation function for each bidder. Let $M$ be an arbitrarily large value. For any bidder $(i,j,l)$ we have:

$$
\val{(i,j,l)}{\sVec'} =
    \begin{cases}
        M^j&  
        \sVec'\in S_{(i,j,l)}\\
        M^j&   
        \sVec' \text{ is obtained by increasing signals from a profile in } S_{(i,j,l)}
        \\
        0 & \text{otherwise}
    \end{cases}
$$
For bidder $i^*$ we have
$$
\val{i^*}{\sVec'} =
    \begin{cases}
        M^{k}&  \sVec'=\sVec
        \\
         M^k&  \sVec' \text{ is obtained by increasing signals from }\sVec
        \\
        0 & \text{otherwise}
    \end{cases}
$$

By the definition we can easily deduce that the valuation is monotone. We need to argue that it is symmetric across expert levels:

\noindent {\bf Expert Symmetry:} We need to argue that this valuation function we define satisfies expert symmetry: 
  Consider any  two signal profiles in
$S_{(i,j,l)}$ and $S_{(i',j',l')}$ respectively. There are three cases:
\begin{itemize}
    \item $l\neq l'$: Clearly the valuation function can differ among these profiles. 
    \item $l=l'$: but $i\neq i'$: The number of type-$l$ experts  with signal $i$ in any strategy profile $S_{(i,j,l)}$ is $i-2$ while there are  $i-1$  such experts in every strategy profile in  $S_{(i',j',l')}$. Therefore the valuation call tell these profiles apart.
    \item $l=l'$, $i=i'$ but $j\neq j'$: Clearly these two signals agree on expertise type and signal count combination apart from $j$ and $j'$ at expert level $l$: the signal from $S_{(i,j,l)}$ has one more signal $j$ and $S_{(i,j',l)}$ has one more signal $j'$.
\end{itemize}

  Note that no signal in $S_{(i,j,l)}$ can be obtained by increasing a signals from a profile in $S_{(i',j',l')}$ if either $i'\neq i$ or $l\neq l'$. Therefore, any bidder $(i',j',l')$ has a positive allocation probability in $S_{(i,j,l)}$ only if $i=i'$ and $l=l'$.  First, bidder $(i,j,l)$ who has valuation $M^j$ for all signals in $S_{(i,j,l)}$. Since signal profiles in $S_{(i,j,l)}$ can be obtained only by increasing signals from profiles in $S_{(i,j',l)}$ for $j'<j$, bidders  $(i,j',l)$ with $j'<j$ have valuation $M^{j'}$ while all the rest have value $0$. As a result, any bidder $(i,j,l)$ is optimal at any signal $ S_{(i,j,l)}$ and any other bidder has at most $1/M$-fraction of the optimal value. Finally, the same is true for signal $\sVec$ where bidder $i^*$ is optimal with value $M^k$ and any other bidder is has at most $1/M$-fraction of it.

Assume that there exists a monotone feasible allocation function that achieves an approximation ratio strictly better than $\nExperts \binom{k}{2} +1$.
Therefore, as $M$ increases it must be the case that the optimal bidder must be allocated with probability strictly higher than $(\nExperts \binom{k}{2} +1)^{-1}$ at each signal profile. As a result, for every bidder $(i,j,l)$ must be allocated with at least that much probability in any signal profile $S_{(i,j,l)}$.
By monotonicity this bidder must also be allocated with at least that much   probability when the signal profile is $\sVec$.
 However this is impossible as total allocation probability at $\sVec$ would exceed $1$ and would violate the assumption that the allocation is feasible.

Finally, any mechanism that is ex-post IR has revenue less than social welfare. As a result, no ex-post IC-IR mechanism obtains revenue more than $\nExperts \binom{k}{2}+1$ fraction of the optimal social welfare.  
 \qed
 \end{proof}

\begin{theorem}\label{thm:genwelfare}
There exists a universally ex-post IC-IR auction 
that achieves a $\nExperts\binom{k}{2} +1$-approximation to the optimal welfare.
 \end{theorem}
  \begin{proof}
 We prove this result using the randomized allocation rule constructed using the following process: $\alloc{i}{\sVec}$ is the probability that we allocate the item to bidder $i$ given the signal profile $\sVec$. 
 
 \begin{algorithm}[H]
\label{alg:allocfunc}
\SetAlgorithmName{Allocation}{}{}
\SetAlgoLined
\DontPrintSemicolon
\SetAlgoNoEnd
 $\rho \gets \min\{n,\nExperts \binom{k}{2}+1\}$\;
 $\alloc{i}{\sVec} \gets 0$ for all $i$ and $\sVec$\;
 \For{each signal profile $\sVec$}{
    $i \gets \opta{\sVec}$\;
    \If{$\alloc{i}{\sVec}=0$}{
    $\alloc{i}{\sVec} \gets 1/\rho$\;
  
     \For{\text{each signal profile $(s'_i, \sVec_{-i})$ with $s'_i > s_i$}}
     {
        $\alloc{i}{s'} \gets 1/\rho$\;
    }}
 }
 \Return $x$\;

\end{algorithm}

 We proceed with proving that this allocation is feasible, approximates the social welfare and finally design an universally IC-IR mechanism by choosing an appropriate pricing scheme.
 
\noindent{\bf Feasibility: } First, we need to show that for any signal profile $\sVec$ the sum of the allocation probabilities does not exceed $1$, i.e., $\sum_{i\in N} x_i(\sVec)\leq 1$.  Each of these bidders' allocation   probability is upper bounded by $(\nExperts \binom{k}{2}+1)^{-1}$. As a result, we complete our proof by bounding the number of bidders with strictly positive allocation by $\nExperts \binom{k}{2}+1$ which implies that our allocation is feasible.  If $ n<  \nExperts \binom{k}{2}+1$ then the result follows trivially. Assume $n>  \nExperts \binom{k}{2}+1$ and consider 
 a bidder and signal $s_i$ in $\sVec$. If this bidder is allocated positive probability it means that the bidder is either the bidder with the maximum value at $\sVec$ (allocated via steps 5 and 6) or was optimal at some signal $(s'_i,\sVec_{-i})$ with $s'_i<s_i$ (allocated via steps 8 and 9 at profile $(s'_i,\sVec_{-i})$). 
 The number of bidders falling to the latter category is at most the number of distinct signals that can be obtained by some bidder reducing their signal. There is $k'-1$ ways of reducing a signal with value $k'$ for every expertise type in $\{1,\dots,\nExperts\}$. Summing over all expertise types and signal levels at $\sVec$ we get

$$
\sum_{\nExperts'=1}^{\nExperts} \sum_{k'=2}^k (k'-1)  = \nExperts \binom{k}{2} 
$$
 
Adding the bidder that is optimal at $\sVec$ we get that the number of bidders with strictly positive probability at $\sVec$ is  $\nExperts \binom{k}{2} +1$.

\noindent{\bf Approximation:} If we set  $\rho = \min\{n,\nExperts{k \choose 2}+1\}$ the allocation is feasible and since at every signal  
profile $\sVec$ we allocate to $\opta{\quality{\sVec}}$ with 
probability $ (\min\{n,\nExperts{k \choose 2}+1\})^{-1}$ the 
approximation ratio follows.

 \noindent{\bf Incentives: } Note  each bidder is allocated with probability   either $0$ or  $1/\rho$. If a bidder is ever allocated an allocation probability $1/\rho$ then the mechanism assigns also an allocation probability for all signal profiles obtained by increasing this bidder's signal in the second loop. This implies that the allocation is monotone. It is well known that we can turn any monotone allocation to a universally ex-post IC-IR mechanism using the payment identity. Since our allocation function assigns allocation probability $0$ or $1/\rho$ at any signal profile the payment function should be: if a bidder $i$ is allocated the item at signal profile $\sVec$ then charge her $\val{i}{\quality{(t,\sVec_{-i}}}$ where $t$ is the smallest sign al  
such that $\alloc{i}{(t,\sVec_{-i}}=1/\rho$.  
\qed
\end{proof}
 
\begin{theorem}\label{thm:genrevenue}
There exists a universally ex-post IC-IR auction 
that obtains revenue that is a $k\cdot (\nExperts\binom{k}{2} +1)$ approximation to the social welfare.
\end{theorem}
  \begin{proof}
  Consider the following mechanism:
 Given a signal $\sVec$, let $i$ be the winner according to allocation 
  defined in the proof of Theorem~\ref{thm:genwelfare} that approximates social welfare. Recall that to obtain a universally IC-IR mechanism we offered the good to the bidder at the price $\val{i}{t,\sVec_{-i}}$ where $t$ is the lowest signal that bidder $i$ can have to win the item with a strictly positive probability. 
  Instead we choose to offer a take-it-or-leave-it  price $\val{i}{\hat{t},\sVec_{-i}}$ where $\hat{t}$ is a random signal in $\{t,\dots,k-1\}$. 
  
  \noindent{\bf Incentives:}
  We argue that this mechanism is universally IC-IR.  We consider two cases
  
  \begin{itemize}
      \item $s_i<t:$ The bidder is never allocated the item. If the bidder raised her signal then price that was offered by the mechanism would be at least  $\val{i}{t,\sVec_{-i}}> \val{i}{\sVec}$. 
      \item $s_i\geq t:$ If the bidder lowers her signal then the item is never allocated and she gets a utility of zero. If the bidder raises her signal then this does not effect the randomized take-it-or-leave it price offered by the mechanism. In any case the deviation is not profitable. 
  \end{itemize}

  \noindent {\bf Revenue:} Conditioned on the optimal bidder being allocated the item in the original allocation the bidder accepts the price 
 that equals her true value $\val{i}{\sVec}$  with probability at least $1/k$. Since these two events are independent with probability 
 $1/(k( \nExperts \binom{k}{2} +1))$ our revenue equals the social welfare. 
 \qed 

 \end{proof}

%
%
%
\bibliographystyle{splncs04}
\bibliography{interdependentbibliography}

\begin{thebibliography}{10}
\providecommand{\url}[1]{\texttt{#1}}
\providecommand{\urlprefix}{URL }
\providecommand{\doi}[1]{https://doi.org/#1}

\bibitem{SOSimproved}
Amer, A., Talgam-Cohen., I.: Auctions with interdependence and {SOS}: Improved
  approximation. In: Proceedings of the 14th International Symposium on
  Algorithmic Game Theory. p. (to appear) (2021)

\bibitem{ausubel1999generalized}
Ausubel, L.M., et~al.: A generalized {V}ickrey auction. Econometrica  (1999)

\bibitem{chawla2014approximate}
Chawla, S., Fu, H., Karlin, A.: Approximate revenue maximization in
  interdependent value settings. In: Proceedings of the fifteenth ACM
  conference on Economics and computation. pp. 277--294 (2014)

\bibitem{constantin2007online}
Constantin, F., Ito, T., Parkes, D.C.: Online auctions for bidders with
  interdependent values. In: Proceedings of the 6th international joint
  conference on Autonomous agents and multiagent systems. pp.~1--3 (2007)

\bibitem{dasgupta2000efficient}
Dasgupta, P., Maskin, E.: Efficient auctions. The Quarterly Journal of
  Economics  \textbf{115}(2),  341--388 (2000)

\bibitem{eden2018interdependent}
Eden, A., Feldman, M., Fiat, A., Goldner, K.: Interdependent values without
  single-crossing. In: Proceedings of the 2018 ACM Conference on Economics and
  Computation. pp. 369--369 (2018)

\bibitem{eden2019combinatorial}
Eden, A., Feldman, M., Fiat, A., Goldner, K., Karlin, A.R.: Combinatorial
  auctions with interdependent valuations: {SOS} to the rescue. In: Proceedings
  of the 2019 ACM Conference on Economics and Computation. pp. 19--20 (2019)

\bibitem{eden2020price}
Eden, A., Feldman, M., Talgam{-}Cohen, I., Zviran, O.: {PoA} of simple auctions
  with interdependent values. In: Thirty-Fifth {AAAI} Conference on Artificial
  Intelligence, {AAAI} 2021. pp. 5321--5329. {AAAI} Press (2021)

\bibitem{ito2006instantiating}
Ito, T., Parkes, D.C.: Instantiating the contingent bids model of truthful
  interdependent value auctions. In: Proceedings of the fifth international
  joint conference on Autonomous agents and multiagent systems. pp. 1151--1158
  (2006)

\bibitem{jehiel2001efficient}
Jehiel, P., Moldovanu, B.: Efficient design with interdependent valuations.
  Econometrica  \textbf{69}(5),  1237--1259 (2001)

\bibitem{klemperer1998auctions}
Klemperer, P.: Auctions with almost common values: The ``wallet game'' and its
  applications. European Economic Review  \textbf{42}(3-5),  757--769 (1998)

\bibitem{krishna2009auction}
Krishna, V.: Auction theory. Academic press (2009)

\bibitem{li2017oxp}
Li, S.: Obvious ex post equilibrium. American Economic Review  \textbf{107}(5),
   230--34 (2017)

\bibitem{li2017obviously}
Li, S.: Obviously strategy-proof mechanisms. American Economic Review
  \textbf{107}(11),  3257--87 (2017)

\bibitem{li2017approximation}
Li, Y.: Approximation in mechanism design with interdependent values. Games and
  Economic Behavior  \textbf{103},  225--253 (2017)

\bibitem{maskin1992auctions}
Maskin, E.: Auctions and Privatization, pp. 115--136. J.C.B. Mohr Publisher
  (1992)

\bibitem{milgrom1982theory}
Milgrom, P.R., Weber, R.J.: A theory of auctions and competitive bidding.
  Econometrica: Journal of the Econometric Society pp. 1089--1122 (1982)

\bibitem{myerson1981optimal}
Myerson, R.B.: Optimal auction design. Mathematics of operations research
  \textbf{6}(1),  58--73 (1981)

\bibitem{robu2013efficient}
Robu, V., Parkes, D.C., Ito, T., Jennings, N.R.: Efficient interdependent value
  combinatorial auctions with single minded bidders. In: Proceedings of the
  23rd International Joint Conference on Artificial Intelligence. p. 339–345.
  IJCAI (2013)

\bibitem{roughgarden2016optimal}
Roughgarden, T., Talgam-Cohen, I.: Optimal and robust mechanism design with
  interdependent values. ACM Transactions on Economics and Computation (TEAC)
  \textbf{4}(3),  1--34 (2016)

\bibitem{surowiecki2005wisdom}
Surowiecki, J.: The Wisdom of Crowds. Anchor (2005)

\bibitem{wilson1987game}
Wilson, R.: Game-theoretic analyses of trading processes. In: Advances in
  Economic Theory, Fifth World Congress. pp. 33--70 (1987)

\bibitem{wilson1969competitive}
Wilson, R.B.: Competitive bidding with disparate information. Management
  Science  \textbf{15}(7),  446--448 (1969)

\end{thebibliography}
%





\end{document}